\documentclass[12pt]{amsart}  
\usepackage{amscd}
\usepackage{amsmath,amssymb,amsthm,yfonts,mathrsfs}

\newtheorem{thm}{Theorem}
\newtheorem{prop}{Proposition}

\newtheorem{cor}{Corollary}

\theoremstyle{definition}

\theoremstyle{definition}
\newtheorem{defn}{Definition}

\newcommand{\End}{\mathrm{End}}

\newcommand{\bR}{\mathbf{R}}

\newcommand{\bC}{\mathbf{C}}
\newcommand{\mT}{\mathcal{T}}

\newcommand{\Ad}{\text{Ad}}
\newcommand{\ad}{\text{ad}}

\newcommand{\mf}{\mathfrak}

\newcommand{\Z}{\mathbb{Z}}

\newcommand{\h}{\mathfrak{h}}
\newcommand{\mk}{\mathfrak{k}}
\newcommand{\n}{\mathfrak{n}}
\newcommand{\be}{\mathfrak{b}}

\numberwithin{equation}{section}

\begin{document}

\title[Classical reflection equation]{Integrable systems from the classical reflection equation}
\author{Gus Schrader}
\address{Department of Mathematics, University of California, Berkeley, CA
94720, USA}
\maketitle

\begin{abstract}
We construct integrable Hamiltonian systems on $G/K$, where $G$ is a coboundary Poisson Lie group and $K$ is a Lie subgroup arising as the fixed point set of a group automorphism $\sigma$ of $G$ satisfying the classical reflection equation.  We show that the time evolution of these systems is described by a Lax equation, and under a factorizability assumption on $G$, present its solution in terms of a factorization problem in $G$.   Our construction is closely related to the semiclassical limit of Sklyanin's integrable quantum spin chains with reflecting boundaries.  
\end{abstract}
\maketitle

\section{Introduction}
A large and well-studied class of integrable Hamiltonian systems consists of those whose phase space can be realized as a Poisson submanifold of a coboundary Poisson-Lie group $G$.  In this situation, the conjugation invariant functions $I_G\subset C(G)$ form a Poisson commutative subalgebra, and particular integrable systems arise by restricting these functions to symplectic leaves in $G$.  

In this paper we construct integrable systems on Poisson homogeneous spaces of the form $G/K$, where $(G,r)$ is a coboundary Poisson Lie group and $K$ is a Lie subgroup of $G$ which arises as the fixed point set of a Lie group automorphism $\sigma:G\rightarrow G$.   In this setting, the condition for $G/K$ to inherit a Poisson structure from $G$ is equivalent to the requirement that the quantity $$C_\sigma(r)=\left(\sigma\otimes\sigma\right)(r)+r-(\sigma\otimes1+1\otimes\sigma)(r)$$ be a $\mathrm{Lie}(K)$-invariant in $\mathfrak{g}\otimes\mathfrak{g}$.  In the special case $C_\sigma(r)=0$, we say that the triple $(\mathfrak{g},r,\sigma)$ is a solution of the {\em classical reflection equation} (CRE).  In this case, we construct a {\em classical reflection monodromy matrix } $\mT$ with the property that the classical reflection transfer matrices obtained by taking the trace of $\mT$ in finite dimensional representations of $G$ form a Poisson commuting family of functions in $C(G/K)\subset C(G)$.  These functions are no longer $\Ad_G$-invariant, but are instead bi-invariant under the action of $K\times K$ on $G$ by left and right translations.  

\bigskip

The motivation for our construction comes from the quantum spin chains with reflecting boundary conditions introduced by Sklyanin \cite{sklyaninref}. It is known \cite{molev} that the algebra of observables of these quantum integrable systems are closely related to coideal subalgebras in the quantum affine algebras $U_q(\widehat{\mathfrak{g}})$.  When $\mathfrak{g}$ is a finite dimensional simple Lie algebra, coideal subalgebras in $U_q(\mathfrak{g})$ have been studied by many authors \cite{noumi},\cite{letzer},\cite{gav},\cite{stokman} and in many cases may be regarded as quantizations of the classical symmetric spaces $G/K$.   Of particular relevance to the present paper is the work of Belliard, Cramp\'{e} and Regelskis \cite{BC},\cite{BR}, who introduced the notion of {\em Manin triple twists} to understand the semiclassical limit of coideal subalgebras.  

\bigskip

We will show that the semiclassical limit of Sklyanin's quantum reflection equation coincides with the CRE for an appropriate choice of group $G$ and automorphism $\sigma$.  We shall also explain how to derive local Hamiltonians for the corresponding homogeneous classical spin chain.  

\bigskip

  The paper is organized as follows.  In section 2, we recall some standard facts about Poisson-Lie groups and their Lie bialgebras. In section 3, we study coideal Lie subalgebras in coboundary Poisson-Lie groups, and introduce the classical reflection equation.   In section 4, we define the reflection monodromy matrix, and show how it can be used to construct a Poisson-commutative subalgebra in $C(G/K)$.    In section 5,  we explain a certain twisting of this construction. Section 6 is devoted to the equations of motion of such systems, which are shown to be of Lax form.  We also provide their explicit solution in the case $G$ admits a certain analog of the Iwasawa decomposition.  The final two sections are devoted to examples.  In section 7, we apply our formalism to the case when $G$ is the split real semisimple Lie group $SL_{n+1}(\bR)$, and $K=SO_{n+1}(\bR)$ is the set of elements fixed by the Cartan involution.  Finally in section 8, we apply our general scheme to the loop algebra $Lgl_2$, and thereby recover the semiclassical limit of a quantum integrable system discovered by Sklyanin, the XXZ spin chain with reflecting boundary conditions.  

\bigskip

{\bf Acknowledgements. } I would like to thank Nicolai Reshetikhin for his suggestion to study the semiclassical limit of spin chains with reflecting boundaries, and for his close reading of several drafts of the text, as well as Alexander Shapiro and Harold Williams for many valuable discussions related to integrable systems and Poisson geometry. 
I would also like to thank L\'aszl\'o Feh\'er and Michael Gekhtman for very helpful correspondence regarding the Toda lattice example treated in Section 7.  
 Finally, I wish to thank the referee for many helpful suggestions, and in particular for pointing out how the Poisson structure on the dual group $G^*$ fits into our framework, as discussed in section 7.  This work was supported by a Fulbright International Science and Technology Award.

\section{Lie bialgebras and Poisson-Lie groups.  }
In this section we collect some standard definitions and facts about Poisson-Lie groups and their Lie bialgebras that we shall use throughout the paper.  For further details, see for example \cite{CP,ES}.  Recall that a {\em Lie bialgebra} is a Lie algebra $\mathfrak{g}$ together with a linear map $\delta: \mathfrak{g}\rightarrow\mathfrak{g}\wedge\mathfrak{g}$ satisfying the following two conditions:
\begin{enumerate}
\item The dual mapping $\delta^*:\mathfrak{g}^*\wedge\mathfrak{g}^*\rightarrow\mathfrak{g}^*$ defines a Lie bracket on $\mathfrak{g}^*$
\item The map $\delta$ satisfies the 1-cocycle condition
$$
\delta([X,Y])=X\cdot\delta(Y) - Y\cdot \delta(X)
$$
where the $\mathfrak{g}$ acts on $\mathfrak{g}\wedge\mathfrak{g}$ by the exterior square of the adjoint representation.
\end{enumerate} 

We shall focus on Lie bialgebras for which the 1-cocycle $\delta$ is actually a coboundary.  This means that there exists an element $r\in \mathfrak{g}\otimes\mathfrak{g}$ such that 
\begin{align}
\label{coboundary}
\delta(X)=X\cdot r
\end{align} 
Note that two elements $r,r'\in\mathfrak{g}\otimes\mathfrak{g}$ that differ by a $\mathfrak{g}$-invariant $\Omega\in \left(\mathfrak{g}\otimes\mathfrak{g}\right)^\mathfrak{g}$ define the same Lie bialgebra structure on $\mathfrak{g}$.   
 One checks \cite{CP} that the induced bracket on $\mathfrak{g}^*$ will be skew and satisfy the Jacobi identity if and only if the symmetric part $J=\frac{1}{2}(r+r_{21})$ of $r$ as well as the quantity
\begin{align}
[[r,r]]:= [r_{12},r_{13}]+[r_{13},r_{23}]+[r_{12},r_{23}]
\end{align}
are $\mathfrak{g}$-invariants in $\mathfrak{g}\otimes\mathfrak{g}$ and $\mathfrak{g}\otimes\mathfrak{g}\otimes \mathfrak{g}$ respectively. If the Lie bialgebra structure on $\mathfrak{g}$ is determined by an element $r$ satisfying $[[r,r]]=0$,  we say $r$ is a solution of the {\em classical Yang-Baxter equation}, and that $(\mathfrak{g},r)$ is a {\em quasitriangular Lie bialgebra.} If $(\mathfrak{g},r)$ is a quasitriangular Lie bialgebra and the symmetric bilinear form on $\mathfrak{g}^*$ defined by $J$ is nondegenerate, we say that $(\mathfrak{g},r)$ is {\em factorizable.}  When $(\mathfrak{g},r)$ is factorizable, the element $J$ establishes an isomorphism $\mathfrak{g}^*\simeq\mathfrak{g}$, which we may use to identify both $J$ itself as well as the skew-part of the $r$-matrix $\hat{r}=\frac{1}{2}(r-r_{21})$ with operators on $\mathfrak{g}$.  Under this identification, these operators satisfy the {\em modified classical Yang-Baxter equation}
\begin{align}
[\hat{r}(X), \hat{r}(Y)]=\hat{r}\left([\hat{X},Y]+[X,\hat{r}(Y)]\right)-[J(X),J(Y)]
\end{align}
\bigskip

A {\em Poisson-Lie group} is Lie group equipped with a Poisson structure such that the group multiplication is a Poisson map. As is well known \cite{CP}, the category of Lie bialgebras is equivalent to the category of connected, simply connected Poisson-Lie groups.  The Lie bialgebra corresponding to a given Poisson-Lie group is called its {\em tangent Lie bialgebra}.  We say that a Poisson-Lie group $G$ is quasitriangular (resp. factorizable) if its tangent Lie bialgebra is.  
 
\bigskip
The Poisson bracket on a coboundary Poisson-Lie group $(G,r)$ may be described quite explicitly.  If $f_1,f_2\in \bC[G]$, we have
\begin{equation}
\label{sklyaninbracket}
\{f_1,f_2\}=\langle df_1\otimes df_2,r^L-r^R\rangle
\end{equation}
where if $X\in\mathfrak{g}$, $X^{L,R}$ denote left/right derivatives with respect to $X$:
$$
X^L[f] (g)= \frac{d}{dt}\bigg|_{t=0}f(e^{tX}g), \ \ X^R[f] (g)= \frac{d}{dt}\bigg|_{t=0}f(ge^{tX})
$$
 
\section{Poisson homogeneous spaces and the classical reflection equation. }
 Let $\mathfrak{g}$ be a Lie bialgebra, and let $\mk\subset\mathfrak{g}$ be a Lie subalgebra in $\mathfrak{g}$.  We say that $\mk$ is a {\em coideal Lie subalgebra} in $\mathfrak{g}$ if
$$
\delta(\mk) \subset \mathfrak{g}\otimes\mk+\mk\otimes\mathfrak{g}
$$
Our interest in coideal Lie subalgebras stems from the following fact, whose straightforward proof may be found in \cite{ES}.  
\begin{prop}
\label{homogeneous}
Let $K\subset G$  be a closed Lie subgroup in $G$. Then the homogeneous space $G/K$ inherits a unique Poisson structure from $G$ such that the natural projection $\pi:G\rightarrow G/K$ is Poisson if and only if $\mk=\mathrm{Lie}(K)$ is a coideal Lie subalgebra in $\mathfrak{g}$.  
\end{prop}
Note that this condition is weaker than the condition that the subgroup $K$ be a Poisson submanifold of $G$, for which we  require $\eta(k)\subset T_kK^{\otimes2}\subset T_kG^{\otimes2}$ for all $k\in K$.  

\bigskip

 It is worthing noting that Proposition \ref{homogeneous} may also be stated in a dual form at the level of function algebras; in this formulation, which goes back to \cite{semenovref}, the subalgebra $\bC[G/K]$ of right $K$-invariant functions on $G$ is a Poisson subalgebra in $\bC[G]$ if and only if the annihilator $\mk^\perp\subset \mathfrak{g}^*$ of $\mk$ is a Lie subalgebra in the dual Lie bialgebra $\mathfrak{g}^*$.

\bigskip

Suppose now that $(\mathfrak{g},r)$ is a coboundary Lie bialgebra, and that $\sigma:\mathfrak{g}\rightarrow\mathfrak{g}$ is a Lie algebra automorphism.   Then the fixed point set $$\mk=\mathfrak{g}^\sigma=\{x\in\mathfrak{g} \big | \sigma(x)=x \}$$ is a Lie subalgebra of $\mathfrak{g}$. Our first goal is to characterize when $\mk$ is a coideal Lie subalgebra.  For this purpose, we shall introduce the quantity 
\begin{equation}
\label{invariant}
C_\sigma(r)=(\sigma\otimes\sigma)(r)+r-\left(\sigma\otimes1+1\otimes\sigma\right)(r)
\end{equation}
Note that if $\Omega\in\mathfrak{g}\otimes\mathfrak{g}$ is any $\mathfrak{g}$-invariant, then $C_\sigma(\Omega)$ is a $\mk$-invariant.  In particular, we have that $C_\sigma(J)$ is a $\mk$-invariant, where $J=\frac{1}{2}(r+r_{21})$ denotes the symmetric part of $r$.  
\begin{thm}
Let $(\mathfrak{g},r)$ be the  coboundary Lie bialgebra determined by an element $r\in\mathfrak{g}\otimes \mathfrak{g}$. Then the Lie subalgebra $\mk=\mathfrak{g}^\sigma$ is a coideal Lie subalgebra in $(\mathfrak{g},r)$ if and only if the quantity $C_\sigma(r)$
is a $\mk$-invariant in $\mathfrak{g}\otimes\mathfrak{g}$.  
\end{thm}
\begin{proof}
The subspace $\mathfrak{g}\otimes\mk+\mk\otimes\mathfrak{g}\subset\mathfrak{g}\otimes\mathfrak{g}$ is the kernel of the vector space endomorphism $A=(\sigma-1)\otimes(\sigma-1)$. So $\mk$ is a coideal Lie subalgebra if and only if $A\circ\delta(\mk)=0$.  Let us expand the $r$-matrix as $r=\sum_i a_i\otimes b_i$ for some $a_i,b_i\in\mathfrak{g}$, and take $x\in\mk$.  Then we compute
\begin{align*}
A\circ\delta(x)&=A([x,a_i]\otimes b_i+a_i\otimes[x,b_i])\\
&=\sigma[x,a_i]\otimes \sigma b_i+[x,a_i]\otimes b_i - \sigma[x,a_i]\otimes b_i -[x,a_i]\otimes \sigma b_i\\
&+\sigma a_i\otimes\sigma [x,b_i]+a_i\otimes[x,b_i]-\sigma a_i\otimes[x,b_i]-a_i\otimes\sigma [x,b_i]\\
&=[x,\sigma a_i]\otimes\sigma b_i+[x,a_i]\otimes b_i-[x,\sigma a_i]\otimes b_i-[x,a_i]\otimes \sigma b_i \\
&+\sigma a_i\otimes [x,\sigma b_i]+a_i\otimes[x,b_i]-\sigma a_i\otimes [x,b_i]-a_i\otimes [x,\sigma b_i]\\
&=[x,(\sigma\otimes\sigma)r]+[x,r]-[x,(\sigma\otimes 1)r]-[x,(1\otimes\sigma)r]\\
&=[x,C_\sigma(r)]
\end{align*}
This is zero if and only if $C_\sigma(r)$ is a $\mk$-invariant.  
\end{proof}
\begin{cor} If $C_\sigma(r)$ is a $\mk$-invariant, the corresponding homogeneous space $G/K$ inherits the structure of a Poisson manifold.  
\end{cor} 
The simplest way to satisfy this condition is to demand that $C_\sigma(r)=0$ for some element $r$ defining the Lie bialgebra structure on $\mathfrak{g}$.  This leads us to the following definition.  
\begin{defn}
A solution of the {\em classical reflection equation} is a triple $(\mathfrak{g},r,\sigma)$ where $\mathfrak{g}$ is coboundary Lie bialgebra, $r\in\mathfrak{g}\otimes\mathfrak{g}$ is an element defining the Lie bialgebra structure on $\mathfrak{g}$ via formula (\ref{coboundary}), and $\sigma$ is a Lie algebra automorphism of $\mathfrak{g}$ such that
\begin{align}
\label{CRE}
(\sigma\otimes\sigma) (r) +r-(\sigma\otimes 1+1\otimes\sigma)(r)=0 
\end{align}
\end{defn}

We will mostly be interested in the case when the automorphism $\sigma$ is an an involution, i.e. $\sigma^2=\mathrm{id}$, or slightly more generally, in the case when $\sigma$ is an automorphism of finite order. When $\sigma$ has finite order, we have the decomposition of $\mk$-modules $\mathfrak{g}=\mk\oplus\mathfrak{p}$, where  $\mathfrak{p}=\oplus _i\mathfrak{p}_i$ is direct sum of the eigenspaces of $\sigma$ with eigenvalues $\lambda_i\neq 1$.    If we decompose the $r$-matrix as
\begin{equation}
\label{decomp}
r=r_{\mk\mk}+\sum_{i}\left(r_{\mk\mathfrak{p}_i}+r_{\mathfrak{p_i}\mk}\right)+\sum_{i,j}r_{\mathfrak{p_i}\mathfrak{p_j}}
\end{equation}
we find $C_\sigma(r)=\sum_{i,j}(\lambda_i-1)(\lambda_j-1)r_{\mathfrak{p_i}\mathfrak{p_j}}$.  Setting $r_{\mf{p}\mf{p}}=\sum_{i,j}r_{\mathfrak{p_i}\mathfrak{p_j}}$, we have
\begin{prop}
If $\sigma$ has finite order, $\mk$ is a coideal Lie subalgebra if and only if $r_{\mf{p}\mf{p}}$ is $\mk$-invariant.  The triple $(\mathfrak{g}, r,\sigma)$ is a solution of the CRE if and only if $r_{\mf{p}\mf{p}}=0$.  
\end{prop}
{\bf Remark. } In this work, we only consider Poisson structures on $G/K$ with the property that the projection $G\rightarrow G/K$ is a Poisson map.  In \cite{drinfeld}, Drinfeld classifies Poisson structures on $G/K$ compatible with the Poisson structure on $G$ in the sense that the mapping $G\times G/K\rightarrow G/K$ is Poisson. Such Poisson structures are shown to correspond to Lagrangian subalgebras in the double $D(\mathfrak{g})=\mathfrak{g}\oplus\mathfrak{g}^*$, those under consideration in the present paper being given by Lagrangians $L_\mk=\mk\oplus\mk^\perp $.  It would be interesting to try to construct integrable systems on the more general class of Poisson homogeneous spaces classified in \cite{drinfeld}.

\bigskip

{\bf Remark. }Let us conclude this section by commenting on the relation between our construction and the one outlined in \cite{BC} in terms of Manin triple twists.  Suppose we have a quasitriangular Lie bialgebra $(\mathfrak{g},r)$, with the corresponding Manin triple $\mathfrak{d}=\mathfrak{g}\oplus \mathfrak{g}^*$, and suppose also that we have a Lie algebra involution $\phi$ of $\mathfrak{g}$. We may then attempt to extend $\phi$ to an anti-invariant Manin triple twist simply by declaring that $\sigma$ acts on $\mathfrak{g}^*$ by
$$(\phi(\xi),X)=-\langle \xi,\phi(X)\rangle$$
where the round brackets to refer to the invariant symmetric bilinear form on $\mathfrak{d}$ and the angle brackets to refer to the canonical pairing of $\mathfrak{g}$ with $\mathfrak{g}^*$.  What must be checked is that this extension respects the Lie algebra structure on $\mathfrak{d}$. In order for this to hold, we require that $(\phi\otimes\phi)(r)+r=0$.  Note, by applying $\phi\otimes1$, that this condition also implies $(\phi\otimes1+1\otimes\phi)(r)=0$.    Therefore, the Manin triple twist constructed in this fashion gives a solution of the classical reflection equation.  The solutions constructed in this fashion have the special property that both `sides' of the reflection equation vanish in their own right.

\section{Construction of integrable systems.  }
Coboundary Poisson-Lie groups play a prominent role in the theory of classical integrable systems because of the following simple consequence of formula (\ref{sklyaninbracket}) for the Poisson bracket on $G$.  

\begin{prop}\cite{RSTS}
If $(G,r)$ is a coboundary Poisson-Lie group, then the subspace $I_G\subset C(G)$ of conjugation-invariant functions is a Poisson commutative subalgebra.  
\end{prop}

Restricting this Poisson commutative subalgebra of functions to symplectic leaves of appropriate dimension in $G$, it is often possible to obtain classical integrable systems.  Examples of integrable systems that can be derived in this framework include the Coxeter-Toda lattice \cite{kolyafour},  its affine counterpart \cite{harold}, and the classical XXZ spin chain with periodic boundaries, see the survey \cite{sixv} and references therein.  We will now explain how to construct integrable systems on the Poisson homogeneous spaces $G/K$ described in the previous section.

In order to describe $K$-invariant functions on $G$ explicitly, we introduce the {\em classical reflection monodromy matrix} \begin{equation}\label{rmm}
\mathcal{T}(g)=g\sigma(g)^{-1}\end{equation}  Observe that if $k\in K=G^\sigma$, we have $\mathcal{T}(gk)=\mathcal{T}(g)$.  Hence matrix elements $\mathcal{T}^V$ of $\mT$ in any finite dimensional representation $V$ are elements of the ring of $K$-invariant functions $C(G/K)$.  Taking the trace, we obtain the {\em reflection transfer matrix} 
\begin{equation}
\label{rtm}\tau^V(g)=tr_V \ \mathcal{T}^V(g)
\end{equation}
In contrast to those arising from the standard construction of integrable systems on $G$, the reflection transfer matrices $\tau^V(g)$ are not in general $\Ad_G$-invariant.  On the other hand, observe that $\tau^V(kg)=\tau^V(g)$, so the reflection transfer matrices lie in $\bC(K\backslash G/K)$, the subalgebra of $K$-bi-invariant functions on $G$.  
\begin{thm}
\label{transfer} Suppose $(\mathfrak{g}, r,\sigma)$ is a solution of the classical reflection equation and $\sigma$ is of finite order. Then the subalgebra $\bC(K\backslash G/K)$ is Poisson commutative.  In particular, for any pair of finite dimensional representations $V,W$ of $G$, the reflection transfer matrices $\tau^V,\tau^W$ satisfy
$$
\{\tau^V,\tau^W\}=0
$$
\end{thm} 
\begin{proof}
If $H\in\bC(K\backslash G/K)$ is a bi-invariant function, then for all $X\in\mk$ we have $X^{L,R}[H]=0$.  Hence the result follows from the decomposition (\ref{decomp}) of the $r$-matrix.  
\end{proof}
Poisson commutativity of the reflection transfer matrices also follows immediately from the following expression for the Poisson brackets of matrix elements of the reflection monodromy matrix:

\begin{prop} Let $(\mathfrak{g},r,\sigma)$ be a solution of the classical reflection equation.  Then matrix elements of the reflection monodromy matrix satisfy
\begin{align}
\label{refpb}
\{\mT_1\otimes\mT_2\}=[r,\mT_1\mT_2]+\mT_1[\mT_2,\sigma_1r]+\mT_2[\mT_1,\sigma_2r]
\end{align}

\end{prop}

\begin{proof}
Recall \cite{CP} that the matrix elements $\rho(g)$ of $g$ in a finite dimensional representation have the Poisson brackets
\begin{align*}
\{\rho_1\otimes \rho_2\}&=[r_{12},\rho_1\rho_2]
\end{align*}
Now since $\sigma$ is a group automorphism, we have
$$
\{\rho_1\circ\sigma,\rho_2\}=[\sigma_1r,(\rho_1\circ\sigma)\rho_2]
$$ et cetera.  
Moreover, since $\mT(g)=\rho(g)\cdot (\rho^{-1}\circ\sigma)(g)$, applying the Leibniz rule yields
\begin{align*}
\{\mT_1\otimes\mT_2\}&=\rho_1\{(\rho_1^{-1}\circ\sigma),\rho_2\}(\rho_2^{-1}\circ\sigma)+\rho_1\rho_2\{\rho_1^{-1}\circ\sigma,\rho_2^{-1}\circ\sigma\}\\
&\ \ +\{\rho_1,\rho_2\}(\rho_1^{-1}\circ\sigma)(\rho_2^{-1}\circ\sigma)+\rho_2\{\rho_1,\rho_2^{-1}\circ\sigma\}(\rho_1^{-1}\circ\sigma)\\
&=r\mT_1\mT_2+\mT_1\mT_2(\sigma_1\sigma_2r)-\mT_1(\sigma_1r)\mT_2-\mT_2(\sigma_2r)\mT_1\\
&+\rho_1\rho_2\left(\sigma_1r+\sigma_2r-r-\sigma_1\sigma_2r\right)(\rho_2^{-1}\circ\sigma)(\rho_1^{-1}\circ\sigma)
\end{align*}
This expression may be rewritten as
\begin{align}
\label{cancel}
\{\mT_1\otimes\mT_2\}&=[r,\mT_1\mT_2]+\mT_1[\mT_2,\sigma_1r]+\mT_2[\mT_1,\sigma_2r]\\ \nonumber
&+\mT_1\mT_2C_\sigma(r)-\rho_1\rho_2C_\sigma(r)(\rho_2^{-1}\circ\sigma)(\rho_1^{-1}\circ\sigma)
\end{align}
Applying the CRE $C_\sigma(r)=0$,  one arrives at formula (\ref{refpb}).  
\end{proof}
We also observe that when $\sigma$ is an involution, the commutativity of the reflection transfer matrices continues to hold under the weaker assumption that $C_\sigma(r)$ is a (possibly nonzero) $\mk$-invariant.  
\begin{prop}
Suppose $\sigma$ is an involution, and $(r,\sigma)$ satisfies the condition (\ref{invariant}): i.e.  $C_\sigma(r)$ is a $\mk$-invariant.  Then the reflection transfer matrices form a Poisson commutative subalgebra.  
\end{prop}
\begin{proof}
As in the previous section, write $\mathfrak{g}=\mk\oplus\mathfrak{p}$ for the decomposition of $\mathfrak{g}$ into the $\pm1$ eigenspaces of $\sigma$.  Recall \cite{helgason} that we may take neighborhoods $V_\mk, V_{\mf{p}}$ in $\mk, \mf{p}$ such that the map $V_\mf{p}\times V_{\mk}\rightarrow G, \ (z,y)\mapsto \exp(z)\exp(y)$ is a diffeomorphism onto an open neighborhood $U$ of the identity in $G$.  For $g=\exp(z)\exp(y)=:pk\in U$, the $K$-invariance of $C_\sigma(r)$ implies
$$
tr_{V\otimes W} \ (g\otimes g)C_\sigma(r)(\sigma(g^{-1})\otimes\sigma(g^{-1}))=tr_{V\otimes W}(p^2\otimes p^2) C_\sigma(r)
$$
Since $g\sigma(g^{-1})=p^2$, the terms involving $C_\sigma(r)$ in (\ref{cancel}) cancel and we obtain the result.  

\end{proof}

\section{Twisting.}
The construction of the previous section also admits a twisted version, which we shall now describe.  Let $(\mathfrak{g},r,\sigma$) be a solution of the classical reflection equation $C_\sigma(r)=0$, and suppose that $\varphi_{\pm}$ are two automorphisms of $\mathfrak{g}$.   As usual, we denote the Lie algebra of fixed points of $\sigma$ by $\mk$, and the corresponding Lie group by $K$.  The automorphisms $\varphi_{\pm}$ allow us to define twisted left and right actions of $K$ on $G$:
\begin{equation}
\label{twistedactions}
k\triangleright g=\varphi_+(k)g, \ \ g\triangleleft k=g\varphi_-(k)
\end{equation}
\begin{prop}
Suppose that the involutions $\sigma_\pm=\varphi_\pm\circ\sigma\circ\varphi_\pm^{-1}$ are again solutions of $C_{\sigma_\pm}(r)=0$.  Then the subalgebra $\bC(K_+\backslash G/K_-)$ of twisted $K$-bi-invariant functions on $G$ is Poisson commutative.  
\end{prop}
\begin{proof}
Since $C_{\sigma_\pm}(r)=0$, we have $r_{\mf{p}_{\pm}\mf{p}_\pm}=0$ in the corresponding decompositions  (\ref{decomp}) of $r$.  The $\sigma_\pm$ fixed-subalgebras $\mk_{\pm}$ are related to $\mk$ by $\mk_{\pm}=\varphi_\pm(\mk)$.   But $f$ is a twisted $K$-bi-invariant function we have
$$
\varphi_+(X)^Lf=0=\varphi_-(X)^Rf
$$ 
for all $X\in\mk$, and from formula (\ref{sklyaninbracket}) for the Poisson bracket the result follows.  
\end{proof}
Of particular interest to us is when the the automorphisms $\varphi_\pm$ are of the form $\Ad_{h_\pm}$ for some $h_\pm\in G$.   In this case, we may form the (right) twisted monodromy matrix
\begin{equation}
\label{twistedmonodromy}
\mT(g)=gh_-\sigma(h_-^{-1}g^{-1})
\end{equation}
whose matrix elements in any finite dimensional representation $V$ are invariant under twisted right action of $K$.  Then the twisted transfer matrices
\begin{equation}
\tau^V=\mathrm{tr}_V \left( \mT(g)\sigma(h_+)h^{-1}_+ \right)
\end{equation}
are twisted bi-invariant functions on $G$.  Putting $K_+=\sigma(h_+)h_+^{-1}$ and $K_-=h_-\sigma(h_-)^{-1}$, we may write the twisted transfer matrix more economically as
\begin{equation}
\label{twistedtransfer}
\tau^V=tr_V \left( gK_-\sigma(g^{-1})K_+\right)
\end{equation}

\section{Factorization dynamics.}
We now proceed to the description of the dynamics of the systems constructed in the previous sections. It is well-known that the flows of $\Ad_G$-invariant Hamiltonians on a factorizable Poisson-Lie group can be described by solving a certain factorization problem in $G$. The first statement of this result in the Poisson-Lie context can be found in \cite{semenovref}; for further details see \cite{RSTS}.  In the present case, we will see that the solution of the reflection dynamics generated by $K$-bi-invariant Hamiltonians is governed by a factorization problem of Iwasawa type.  

\bigskip

 For simplicity, we shall work in the untwisted setting.   Let us begin by writing down the equations of motion generated by $K$-bi-invariant Hamiltonians on a coboundary Poisson-Lie group $(G,r)$.  

Given a function $H\in\bC(G)$, its left and right gradients at a point $g\in G$ are functionals $\nabla^{\pm}H(g)\in\mathfrak{g}^*$ defined by
$$
\langle \nabla^+H(g),X\rangle=\frac{d}{dt}\bigg|_{t=0}H(e^{tX}g), \ \ \langle \nabla^-H(g),X\rangle=\frac{d}{dt}\bigg|_{t=0}H(ge^{tX})
$$
In this section, we shall assume that we have fixed a finite dimensional representation $(\rho,V)$ of $G$, and to simplify notation we shall confuse group elements $g\in G$ with their images $\rho(g)\in \End(V)$.  Now, it follows from formula (\ref{sklyaninbracket}) that the matrix $g$ evolves under the Hamiltonian flow of $H$ by 
\begin{equation}
\label{geneq}
\dot{g}(t)=r\left(\nabla^+H\right)g-gr\left(\nabla^-H\right)
\end{equation} 
 where we regard the tensor $r\in\mathfrak{g}\otimes\mathfrak{g}$ as a linear map $\mathfrak{g}^*\rightarrow\mathfrak{g}$ by contraction in the first tensor factor.  

\begin{prop} Suppose that $(G, r,\sigma)$ is a solution of the classical reflection equation with $\sigma$ of finite order, and that $H\in C(K\backslash G/K)$ is a $K$-bi-invariant Hamiltonian.  Then the Hamiltonian flow of $H$ takes place on $K\times K$-orbits in $G$, and the reflection monodromy matrix evolves in time by the Lax equation
\begin{equation}
\label{lax}
\dot{\mT}(t)=[r\left(\nabla^+H\right)(t),\mT(t)]
\end{equation}
\end{prop}
\begin{proof}
If $H\in\bC(K\backslash G/K)$, then for $X\in\mk$, $\langle\nabla^{\pm}H ,X\rangle=0$, and so from the decomposition (\ref{decomp}) for $r$, it follows that  $r_{\mathfrak{p}\mathfrak{p}}=0$ and $r(\nabla^{\pm}H)\in\mk$.  In view of the equation of motion (\ref{geneq}), this proves the first part of the proposition.  The equation of motion for the reflection monodromy matrix is obtained by straightforward calculation using formula (\ref{sklyaninbracket}).  
\end{proof}

Let us now assume that $(G,\tilde{r})$ is a quasitriangular Poisson-Lie group, with $\tilde{r}$ satisfying the classical Yang-Baxter equation in the form
\begin{align}
\label{cyb}
 [\tilde{r}_{12},\tilde{r}_{13}]+[\tilde{r}_{13},\tilde{r}_{23}]+[\tilde{r}_{12},\tilde{r}_{23}]=0
\end{align}
We will further suppose that there exists an element $r\in\mathfrak{g}\otimes \mathfrak{g}$ such that $\Omega:=\tilde{r}-r\in(\mathfrak{g}\otimes\mathfrak{g})^\mathfrak{g}$, so that $r$ and $\tilde{r}$ define the same Lie bialgebra structure on $\mathfrak{g}$, and such that $(\mathfrak{g},r,\sigma)$ is a solution of the classical reflection equation $C_\sigma(r)=0$ for a finite order automorphism $\sigma$. By the Yang-Baxter condition (\ref{cyb}), the linear map \begin{align}
 \tilde{r}:\mathfrak{g}^*\rightarrow \mathfrak{g}, \ \ \xi\longmapsto \langle \xi\otimes 1, r\rangle
 \end{align}
  is a homomorphism of Lie algebras, so its image $\be=\tilde{r}(\mathfrak{g}^*)$ is a Lie subalgebra in $\mathfrak{g}$. 
  
  \bigskip
  
  In order to reduce the reflection dynamics to a factorization problem in $G$, we shall need to make the assumption that $\mathfrak{g}$ admits an {\em Iwasawa decomposition} $\mathfrak{g}=\be\oplus\mk$.  At the group level, this means that in a neighborhood of the identity, each element of $g$ admits a unique factorization $g=bk^{-1}$ with $b\in B$, $k\in K$.  
\begin{prop} Under the above assumptions on $(\mathfrak{g},r,\tilde{r},\sigma)$, the time evolution of the matrix $g(t)$ under the Hamiltonian flow of $H\in\bC(K\backslash G/K)$ is given for sufficiently small time $t$ by
\begin{equation}
\label{facdynamics}
g(t)=k_+^{-1}(t)g_0k_-(t)
\end{equation}
where the matrices $k_\pm(t)$ are solutions of the following factorization problems in $G$:
$$\exp\left(t\Omega(\nabla^{\pm}H(g_0))\right)=b(t)k_{\pm}^{-1}(t)$$ 
\end{prop}

\begin{proof} For brevity, let us denote $Q_\pm=\Omega(\nabla^\pm H(g_0))$.  By the Iwasawa decomposition, for sufficiently small time we have a unique factorization
$$
e^{tQ_\pm }=b(t)k_\pm^{-1}(t)
$$
where $b(t)\in B, k_\pm(t)\in K$.  Differentiating with respect to time shows that
$$
b^{-1}\dot{b}-k_\pm^{-1}\dot{k_\pm}=k_\pm^{-1}Q_\pm k_\pm
$$
Now set
$
g(t)=k_+^{-1}(t)g_0k_-(t)
$.  
By the bi-invariance of $H$ and the $\Ad_G$-invariance of $\Omega$, we have 
\begin{align*}
k_\pm^{-1}(t)Q_\pm k_\pm(t)&=\Omega\left(\nabla_{g(t)}^\pm H\right)\\
&=\tilde{r}\left(\nabla_{g(t)}^\pm H\right)-r\left(\nabla_{g(t)}^\pm H\right)
\end{align*}
But by the Iwasawa decomposition at the Lie algebra level this implies
$$
k_\pm^{-1}\dot{k_\pm}=r\left(\nabla_{g(t)}^\pm H\right)
$$
and by comparison with the equation of motion (\ref{geneq}) the result follows.  
\end{proof}

\begin{cor} Under the above assumptions on $(\mathfrak{g},r,\tilde{r},\sigma)$, the isospectral evolution of the reflection monodromy matrix under the Hamiltonian flow of $H\in\bC(K\backslash G/K)$ is given explicitly by \begin{equation}
\label{facdynamics}
\mT(t)=k_+^{-1}(t)\mT_0k_+(t)
\end{equation}
\end{cor}
In particular, it follows that all spectral invariants of the reflection monodromy matrix are conserved quantities under the flow of the reflection Hamiltonians.

\section{Finite dimensional examples. }
In this section we shall consider some examples of our construction applied to finite dimensional Lie algebras.   We begin with an example that, while not suited for the construction of integrable systems, is nonetheless important from the point of view of duality theory of Poisson-Lie groups, see \cite{semenovref}, \cite{qdp}. 

\bigskip

 Suppose that $(\mathfrak{g},r)$ is a factorizable Lie bialgebra, and as usual let $J=\frac{1}{2}(r+r_{21})$, $\hat{r}=\frac{1}{2}(r-r_{21})$ be the decomposition of $r$ into its symmetric and skew parts.    We write $r_+,r_-:\mathfrak{g}^*\rightarrow \mathfrak{g}$ for the homomorphisms of Lie algebras determined by contraction with $r$ and $-r_{21}$ respectively in the first tensor factor.  Let $\mathfrak{d}$ be the double Lie bialgebra of $\mathfrak{g}$. As a Lie algebra, $\mathfrak{d}\simeq \mathfrak{g}\oplus\mathfrak{g}$, with $\mathfrak{g}$ being embedded as the diagonal subalgebra $\mathfrak{g}_\Delta\subset\mathfrak{g}\oplus\mathfrak{g}$, and $\mathfrak{g}^*$ being embedded via the homomorphism of Lie algebras $r_+\oplus r_-:\mathfrak{g}^*\rightarrow \mathfrak{g}\oplus\mathfrak{g}$.  The bialgebra structure on $\mathfrak{d}$ is determined by the $r$-matrix
$$
r_\mathfrak{d}=\sum_{i}(x_i,x_i)\otimes (r_+(\xi_i),r_-(\xi_i))
$$
where $x_i,\xi_i$ are dual bases in $\mathfrak{g},\mathfrak{g}^*$.  Now consider the Lie algebra involution \begin{align}
\sigma:\mathfrak{d}\rightarrow \mathfrak{d}, \ \ \ \sigma(x,y)=(y,x)
\end{align}
whose fixed point set is precisely the Lie sub-bialgebra $\mathfrak{g}_\Delta$.  
Since $(\sigma\otimes 1)r_\mathfrak{d}=r_\mathfrak{d}$, the triple $(\mathfrak{d},r_\mathfrak{d},\sigma)$ is a solution of the classical reflection equation.  At the group level, the reflection monodromy map (\ref{rmm}) is given by
\begin{align}
\label{projectdual}
\mathcal{T}: (G\times G)/G_\Delta\rightarrow (G\times G), \ \ (g_1,g_2)\longmapsto (g_1g_2^{-1}, g_2g_1^{-1})
\end{align}
Composing with the projection onto the first factor yields a map $(G\times G)/G_\Delta\rightarrow G$, and the quotient Poisson structure on $G$ is identified with the pushforward of that of the dual Poisson-Lie group $G^*\subset G\times G$.  Moreover, if $V$ is a representation of $G$, then tensoring with the trivial representation gives a representation $V\otimes\bC$ of $G\times G$, and evaluating formula (\ref{refpb}) for the Poisson brackets of reflection monodromy matrix elements in this representation yields the well-known formula
\begin{align}
\label{dualbrack}
\{\mathcal{T}_1\otimes\mathcal{T}_2\}&= r_{12}\mT_1\mT_2-\mT_1\mT_2r_{21}-\mT_1r_{12}\mT_2+\mT_2r_{21}\mT_1\\
&=\hat{r}T_{1}T_{2}+T_1T_2\hat{r}-T_1(\hat{r}+J)T_2-T_2(\hat{r}-J)T_1
\end{align}
for the pushforward to $G$ of the Poisson bracket of the dual group $G^*$, as found in \cite{semenovref}.   The $G_\Delta$-bi-invariant reflection Hamiltonians are identified with the conjugation invariant functions on $G$, which are Casimirs with respect to the bracket (\ref{dualbrack}) and thus generate trivial Hamiltonian dynamics on $G$.  
\bigskip

To obtain examples with nontrivial dynamics, we shall now apply our general construction to the case of split real semisimple Lie algebras.  For simplicity, we will focus on the case of type $A_n$, when $\mathfrak{g}=sl_{n+1}(\bR)$.  

\bigskip

Let us choose a triangular decomposition $\mathfrak{g}=\n^-\oplus\h\oplus\n^+$, and a system of Chevalley generators $\{E_i,F_i,H_i\}$. We denote the set of positive roots by $\Delta^+$.    The standard Lie bialgebra structure on $\mathfrak{g}$ is defined by
\begin{align}
\label{cobracket}
\delta(H_i)=0, \ \ \delta(E_i)=E_i\wedge H_i, \ \  \delta(F_i)=F_i\wedge H_i,
\end{align}
or equivalently by the skew $r$-matrix
\begin{align}
\hat{r}=\sum_{\alpha\in\Delta^+}E_\alpha\wedge F_\alpha
\end{align}
The {\em Cartan involution} on $\mathfrak{g}$ is the Lie algebra automorphism $\theta$ defined by 
$$
\theta(H_i)=-H_i, \ \ \theta(E_i)=-F_i, \ \  \theta(F_i)=-E_i
$$
This involution gives rise to a decomposition $\mathfrak{g}=\mathfrak{k}\oplus\mathfrak{p}$ into its $\pm 1$ eigenspaces known as the {\em Cartan decomposition}. The fixed point set $\mathfrak{k}=so_n(\bR)$ is a Lie subalgebra in $\mathfrak{g}$, the anti-fixed point set $\mathfrak{p}$ consisting of traceless symmetric matrices is a $\mathfrak{k}$-module. 
\begin{prop} The triple $(\mathfrak{g}, \hat{r},\theta)$ is a solution of the classical reflection equation.  
\end{prop}
From this (or by inspection of the formulas for the cobracket), it follows that $\mathfrak{k}$ is a coideal Lie subalgebra in $\mathfrak{g}$.    Note also that $\mathfrak{k}$ does {\em not} satisfy $\delta(\mathfrak{k})\subset\mathfrak{k}\otimes\mathfrak{k}$, so $K$ is not a Poisson-Lie subgroup in $G$.  

\bigskip

 Denote by $A,B$ the analytic subgroups of $G$ with Lie algebras $\h, \mathfrak{b}=\h\oplus\mathfrak{n}$ respectively.  By virtue (see \cite{helgason}) of the Iwasawa decomposition $G=BK$, we may identify the homogeneous space $G/K$ with $B$.  Observe that $B\subset G$ is a Poisson submanifold.  
\begin{prop}
The Poisson structure on $G/K$ coincides with the Poisson structure on $B$ coming from its inclusion into $G$.  
\end{prop}
\begin{proof}
Let $\pi:G\rightarrow G/K\simeq B$ be the projection with respect to Iwasawa decomposition, and $\iota:G/K\simeq B\rightarrow G$ be inclusion.  We must show that $\iota$ is Poisson, where $B\simeq G/K$ is equipped with the {\em quotient} Poisson structure.  For $b\in B$, we calculate $\iota_*^{\otimes2}\left(\eta_{G/K}(bK)\right)$.  By definition of the quotient Poisson structure and the chain rule, we have
\begin{align*}
\iota_*^{\otimes2}\left(\eta_{G/K}(bK)\right)&=\iota_*^{\otimes2}\pi_*^{\otimes2}\left(\eta_{G}(b)\right)\\
&=(\iota\circ\pi)_*^{\otimes2}\left(\eta_{G}(b)\right)
\end{align*}
Since $B\subset G$ is Poisson, $\eta_{G}(b)\subset TB\otimes TB$.  Moreover, we have $(\iota\circ\pi)|_{B}=\text{id}_B$.  Hence $$(\iota\circ\pi)_*^{\otimes2}\left(\eta_{G}(b)\right)
=\eta_{G}(b)$$ which shows that $\iota$ is Poisson.  
\end{proof} 
{\bf Remark. } Recall \cite{helgason} that we also have the {\em global Cartan decomposition} $G=PK$: any element of $g$ may be uniquely factored $g=pk$, where $k\in SO_n$ and $p$ is a symmetric positive definite matrix.  This shows that we may also identify $G/K$ with the space $P$ of symmetric positive definite matrices.  The element $p$ in the factorization of $g$ is essentially the reflection monodromy matrix, since it may be explicitly computed as $p^2=gg^T=g\theta(g^{-1})$.  

\bigskip
%
By the so-called $KAK$ decomposition (see again \cite{helgason}) , the reflection Hamiltonians may be regarded as functions on $K\backslash G/K\simeq A$, where $A$ is the $n$-dimensional Cartan subgroup of unit determinant diagonal matrices. 
Thus in order to describe integrable systems on $G/K$, we must identify symplectic leaves of dimension $2n$.   Recall \cite{kolyafour}, \cite{harold} that the double Bruhat cells $G^{u,v}$ are $A$-invariant Poisson subvarieties of $G$.  Consider the $2n$-dimensional double Bruhat cell $G^{c,1}\subset B$ where $c$ is the Coxeter element $c=s_n\cdots s_1$ in the symmetric group  $S_{n+1}$. 
\begin{prop}
The double Bruhat cell $G^{c,1}$ is mapped isomorphically under the quotient projection to a symplectic manifold $M_c\subset G/K$ of dimension $2n$, to and the restriction of $C(K\backslash G/K)$ to $M_c$ defines an integrable system. 
\end{prop}
Indeed,  $G^{c,1}$ consists of all upper triangular matrices $X$ with positive diagonal entries and with all entries of distance $>1$ from the diagonal equal to zero. If we write $a_k=X_{kk},k=1,\ldots, n+1, \ b_k=X_{k,k+1}$, the non-zero Poisson brackets of coordinates are given by
$$
\{a_{k},b_{k}\}=a_{k}b_{k}, \ \ \{a_{k+1},b_{k}\}=-a_{k+1}b_{k}
$$
The coordinates $(a_k,b_k)$ can be expressed in terms of canonically conjugate coordinates $\{p_k,q_k\}=1,\  k=1,\ldots n$ by
$$
a_{k}=e^{q_{k-1}-q_{k}}, \ \ b_{k}=e^{p_k}, \ \ k=1,\ldots, n
$$
where we understand $q_{0}=q_{n+1}=0$.  The reflection monodromy matrix takes the symmetric tridiagonal form
\begin{align}
\label{linref}
\mT=\left[\begin{array}{cccccc}a_1^2+b_{1}^2&b_{1}a_2 &0&0&0&0\\
b_{1}a_{2} &a_{2}^2+b_{2}^2&b_{2}a_{3}&0&0&0\\
0&b_{2}a_{3}&a_{3}^2+b_{3}^2&b_{3}a_{4}&0&0\\
&&\ddots&&&\\
&&&\ddots&&\\
0&0&0&b_{n-1}a_n&a_{n}^2+b_n^2&b_na_{n+1}\\
0&0&0&0& b_n a_{n+1} &a_{n+1}^2
\end{array}\right]
\end{align}
Its trace is the quadratic local Hamiltonian 
$$
tr \ \mT= \sum_{k=1}^{n+1}e^{2(q_{k-1}-q_{k})}+\sum_{k=1}^ne^{2p_k}
$$


This system essentially coincides with the open Coxeter-Toda system with phase space $G^{c,c^{-1}}/A$, as described in \cite{GSV}.  Explicitly, $G^{c,c^{-1}}/A$ consists of the unit determinant tridiagonal matrices, modulo conjugation by diagonal matrices, where we equip $SL_{n+1}(\bR)$ with the Poisson structure defined by the scaled $r$-matrix $2r$.     Then under this normalization, the map $M_c\rightarrow G^{c,c^{-1}}/A, \ \ \mT\mapsto [\mT]\in SL_{n+1}/A$ is Poisson, and carries open Coxeter-Toda Hamiltonians to reflection Hamiltonians.  Hence, our construction gives a ``symmetric" Lax representation of the open Coxeter-Toda system.  This result is a non-linear analog of the fact that  the phase space of the non-relativistic open Toda chain with its linear Poisson structure may be realized in two ways: either as lower Hessenberg matrices, or as symmetric tridiagonal matrices, see \cite{quantumclassical} and \cite{gekhtmanshapiro}. This linear Poisson structure on the symmetric tridiadonal matrices can be obtained by a making an appropriate linearization of (\ref{linref}) in the neighborhood of the identity matrix, see \cite{kolyafour} Section 7.2.  
\bigskip

One may also restrict the reflection Hamiltonians to the $2n$-dimensional symplectic leaves of the double Bruhat cells $G^{c,c}$, although we have been unable to explicitly identify the integrable systems obtained in this fashion with those in the existing literature.  

\bigskip

{\bf Remark. } Although the restriction of the reflection Hamiltonians to symplectic leaves of dimension greater than $2n$ cannot yield an integrable system, the explicit solution of the equations of motion in terms of the factorization problem in $G$ given in Section 6 remains valid on such leaves.  This leads us to suspect that such systems may be {\em degenerately integrable}, as their $\Ad_G$-invariant counterparts were shown to be in \cite{degen}.  We leave the detailed investigation of this subject for a future work.

\section{Classical XXZ spin chain with reflecting boundaries. }
We will now apply our general scheme to the case of the formal loop algebra $Lgl_2=gl_2\otimes\bC[z^{\pm1}]$.  In doing so, we shall recover the semiclassical limit of Sklyanin's XXZ model with reflecting boundary conditions \cite{sklyaninref}.   Let us choose the basis
$$
E=\left[\begin{array}{cc}0&1\\0&0 \end{array}\right], \ \  H=\left[\begin{array}{cc}1&0\\0&-1 \end{array}\right], \ \ F=\left[\begin{array}{cc}0&0\\1&0 \end{array}\right], \ \ I=\left[\begin{array}{cc}1&0\\0&1 \end{array}\right]
$$
for the Lie algebra $\mathfrak{g}=gl_2$.  
The infinite dimensional Lie algebra $L\mathfrak{g}$ has a basis  $\{x[n]=x\otimes z^n \  \big | x\in \{E,H,F,I\}, n\in\Z\}$.  It admits several pseudo-triangular Lie bialgebra structures \cite{CP}, the one of interest to us being determined by the trigonometric $r$-matrix

\begin{align}
\label{trigrmatrix}
r(z,w)=\frac{z^2+w^2}{z^2-w^2}\left(\frac{H\otimes H}{2}\right) +\frac{2zw}{z^2-w^{2}}\left(E\otimes F+F\otimes E\right)
\end{align}
Here we regard $L\mathfrak{g}\otimes L\mathfrak{g}\simeq(\mathfrak{g}\otimes\mathfrak{g})[z^{\pm1},w^{\pm1}]$ as being embedded in the larger space $(\mathfrak{g}\otimes\mathfrak{g})(z,w)$ of $\mathfrak{g}\otimes\mathfrak{g}$-valued rational functions of $z$ and $w$.   The cobracket is given by the formula
$$
\delta(x)(z,w)=\left(\ad_{x(z)}\otimes1+1\otimes\ad_{x(w)}\right)r(z,w)
$$
Setting $r_{12}(z/w)=r(z,w)$, we obtain a solution of the classical Yang-Baxter equation in $(\mathfrak{g}\otimes\mathfrak{g})(z,w)$ with multiplicative spectral parameter:
\begin{equation*}
[r_{12}(z),r_{13}(zw)]+[r_{12}(z),r_{23}(w)]+[r_{13}(zw),r_{23}(w)]=0
\end{equation*}
Observe 
that $r(z)$ satisfies the `unitarity' condition
$$
r_{12}(z^{-1})=-r_{12}(z)
$$
Now let $\rho:gl_2\rightarrow\End(\bC^2)$ be the vector representation of $gl_2$, and consider evaluation representations $\rho_z=\bC^2\otimes\bC[z^{\pm1}],\rho_w=\bC^2\otimes\bC[w^{\pm1}]$.  Again, we embed $\rho_z\otimes\rho_w$ inside $(\bC^2\otimes\bC^2)(z,w)$.   Then the image of $r(z,w)$ is
\begin{align*}
r(z,w)&=\frac{1}{2(z^2-w^2)}\left[\begin{array}{cccc}z^2+w^2&0&0&0\\
						0&-(z^2+w^2)&4zw&0\\
						0&4zw&-(z^2+w^2)&0\\
						0&0&0&z^2+w^2
 \end{array}\right]\end{align*}

The classical monodromy matrix $T(z)$ is defined as the matrix elements of $LGL_2$ in the evaluation representation $\rho_z$.  Poisson brackets of its elements are given by
\begin{align}
\{T_1(z),T_2(w)\}=[r_{12}(z,w),T_1(z)T_2(w)]
\end{align}
The loop algebra $L\mathfrak{g}$ has an involution $\theta$ defined by $(\theta x)(z)=x(z^{-1})$ for $x(z)\in L\mathfrak{g}$. Conjugation by any element of the loop group $LGL_2$ also defines an automorphism of $L\mathfrak{g}$.  In order to reproduce the diagonal solution of the reflection equation studied by Sklyanin in \cite{sklyaninref}, we introduce a one-parameter family of loop group elements $h(z;\xi)$ defined by
\begin{equation}
h(z;\xi)=\left(\begin{array}{cc}1 &0\\0&\xi z^{-1}-z\xi^{-1} \end{array}\right)
\end{equation}
We can then consider the composite automorphism $\sigma_\xi=\Ad_{h(z;\xi)}\circ\theta\circ\Ad_{h^{-1}(z;\xi)}$. Note that setting $\xi=i$ recovers $\sigma_{i}=\theta$.  
\begin{prop} 
\label{cantwist} We have
\begin{equation}
\label{abstractref}
(\sigma_\xi\otimes\sigma_\xi)r(z,w)+r(z,w)-(\sigma_\xi\otimes1+1\otimes\sigma_\xi)r(z,w)=0
\end{equation}
so that $\sigma_\xi$ defines a one-parameter family of solutions of the classical reflection equation.  Moreover, the diagonal solution of the reflection equation considered by Sklyanin in \cite{sklyaninref} is obtained as
\begin{equation*}
\mathcal{K}(z;\xi)=h(z;\xi)h^{-1}(z^{-1};\xi)=\left(\begin{array}{cc}1 &0\\0&\frac{\xi z^{-1}-\xi^{-1}z}{\xi z-z^{-1}\xi^{-1}} \end{array}\right)
\end{equation*}
and in terms of the matrix $\mathcal{K}(z;\xi)$, equation (\ref{abstractref}) becomes\begin{align*}
&r_{12}(z/w)\mathcal{K}_1(z)\mathcal{K}_2(w)+\mathcal{K}_1(z)r_{12}(zw)\mathcal{K}_2(w)\\
&=\mathcal{K}_2(w)r_{12}(zw)\mathcal{K}_1(z)+\mathcal{K}_2(w)\mathcal{K}_1(z)r_{12}(z/w)
\end{align*}
where we for brevity we have suppressed in our notation the $\xi$-dependence of the matrix $\mathcal{K}(z)$.  
\end{prop}

{\bf Remark. } Formula (\ref{abstractref}) can be recognized as the semiclassical limit of Sklyanin's quantum reflection equation \cite{sklyaninref}
\begin{equation*}
R_{12}(z_1/z_2)\mathcal{K}_1(z_1)R_{12}(z_1z_2)\mathcal{K}_2(z_2)=\mathcal{K}_2(z_2)R_{12}(z_1z_2)\mathcal{K}_1(z_1)R_{12}(z_1/z_2)
\end{equation*}
where the quantum $R$-matrix
\begin{align*}
R(z)=\left[\begin{array}{cccc}1&0&0&0\\
						0&b(z)&c(z)&0\\
						0&c(z)&b(z)&0\\
						0&0&0&1
 \end{array}\right] \ \ \text{for} \ \ b(z)=\frac{z-z^{-1}}{qz-q^{-1}z^{-1}}, \ \ c(z)=\frac{q-q^{-1}}{qz-q^{-1}z^{-1}}
\end{align*}
is related to $r(z)$ by $$f(z)R(z)=1+hr(z)+O(h^2), \ q=e^{h}, \ h\rightarrow 0$$ where $$f(z)=\frac{q^{1/2}z-q^{-1/2}z^{-1}}{z-z^{-1}}$$

\bigskip
By virtue of Proposition \ref{cantwist}, we may perform the twisting outlined in section 5.   If $\xi_+,\xi_-$ are complex numbers, we shall twist on the left by $\Ad_{h(z;\xi^{-1}_+)}$, and on the right by $\Ad_{h(z;\xi_-)}$. The corresponding twisted reflection monodromy matrix is $$\mT(z)=T(z)\mathcal{K}_-(z)T^{-1}(z^{-1})$$  A simple calculation along the lines of the proof of formula (\ref{refpb}) shows that
\begin{align*}
\{\mT_1(z)\otimes\mT_2(w)\}=[r_{12}&(z/w),\mT_1(z)\mT_2(w)]\\
& + \mT_1(z)r_{12}(zw)\mT_2(w)-\mT_2(w)r_{12}(zw)\mT_1(z)
\end{align*}
which coincides with the formula for the Poisson bracket of reflection monodromy matrix elements given in Sklyanin's original paper \cite{sklyaninref}.    

\bigskip

In order to describe particular finite dimensional systems, we must identify symplectic leaves in $LGL_2$. The leaves we will consider can be described in terms of $\left(SL_2\right)^*$, the Poisson-Lie group dual to $SL_2$ with its standard Poisson structure.

Let $B_\pm$ be the standard pair of opposite Borel subgroups in $SL_2$, and let $\pi_\pm:B_\pm\rightarrow \mathcal{H}$ be the natural projections to the Cartan subgroup $\mathcal{H}\subset SL_2$.  Recall that the Poisson-Lie group $SL_2^*$ can be realized as the subset in $B_+\times B_-$ consisting of pairs $(x,y)$ with $\pi_+(x)\pi_-(y)=1$. As such, we may coordinatize $SL_2^*$ as the set of pairs of matrices
\begin{align}
SL_2^*=\left\{\left(\left[\begin{array}{cc}k&e\\0&k^{-1} \end{array}\right],\left[\begin{array}{cc}k^{-1}&0\\f&k \end{array}\right]\right) \ \Bigg{|} \  k\in \bC^*, e,f\in \bC\right\}
\end{align} 
The Poisson brackets of the coordinate functions $e,f,k$ are given by
\begin{align}
\label{genpb}
\nonumber &\{k,e\}= ke\\ & \{k,f\}=- kf \\ & \{e,f\}=2(k^{2}-k^{-2})\nonumber
\end{align}
The function $\omega=k^2+k^{-2}+ef$ is a Casimir element of the Poisson algebra $\bC[SL_2^*]$, and it is thus constant on symplectic leaves. We shall parameterize its value by $\omega=t^2+t^{-2}$.   The generic level set $\omega_t$ is a two-dimensional symplectic leaf $\Sigma_t$ in $SL_2^*$.  

\bigskip

The relation of the Poisson manifold $SL_2^*$ to the loop group $LGL_2$ arises from the observation that the Poisson brackets (\ref{genpb}) can be recast in matrix form by introducing 
$$
L(z)=\left(\begin{array}{cc} zk-z^{-1}k^{-1} & e\\f&zk^{-1}-z^{-1}k  \end{array}\right)
$$

\begin{prop}
The matrix $L(z)$ satisfies \begin{align}\{L_1(z)\otimes L_2(w)\}=[r_{12}(z/w),L_1(z)L_2(w)]\end{align} Hence the mapping
\begin{align}
\varphi: & \ SL_2^*\longrightarrow LGL_2,\\
\nonumber 
&\left(\left[\begin{array}{cc}k&e\\0&k^{-1} \end{array}\right],\left[\begin{array}{cc}k^{-1}&0\\f&k \end{array}\right]\right)\longmapsto L(z)
\end{align}
defines a Poisson embedding of $SL_2^*$ into the loop group $LGL_2$.  
\end{prop}
{\bf Remark. } The Poisson map $\varphi$ is the semiclassical counterpart of the {\em evaluation homomorphism} $U_q(\widehat{gl_2})\rightarrow U_q(gl_2)$ which is fundamental in the study of finite dimensional representations of the quantum affine algebra $U(\widehat{gl_2})$, see \cite{CP}.  

\bigskip

Now since group multiplication $LGL_2\times LGL_2\rightarrow LGL_2$ in the Poisson-Lie group $LGL_2$ is Poisson, $\varphi$ gives rise to a Poisson map
\begin{align}
&\varphi^{(N)}:  \  \left(SL_2^*\right)^{\times N}\longrightarrow LGL_2,\\
\nonumber 
&\left(\left[\begin{array}{cc}k_n&e_n\\0&k_n^{-1} \end{array}\right], \left[\begin{array}{cc}k_n^{-1}&0\\f_n&k_n\end{array}\right]\right)^N_{n=1}\longmapsto L_1(z)\cdots L_N(z)
\end{align}
where 
$$
L_n(z)=\left(\begin{array}{cc} zk_n-z^{-1}k_n^{-1} & e_n\\f_n&zk_n^{-1}-z^{-1}k_n  \end{array}\right)
$$
Restricting $\varphi^{(N)}$ gives an embedding of the $2N$-dimensional symplectic manifold $\Sigma_{\mathbf{t}}=\Sigma_{t_1}\times\cdots\times\Sigma_{t_N}$ into $LGL_2$, whose image is the phase space of the classical XXZ spin chain with $N$ sites.  The reflection monodromy matrix then becomes
\begin{align}
\mT(z)=L_1(z)\cdots L_N(z)\mathcal{K}_-(z)L_{N}^{-1}(z^{-1})\cdots L_{1}^{-1}(z^{-1})
\end{align}

\bigskip

Let us conclude by deriving the local Hamiltonian of the homogeneous chain where $\omega_i\equiv \omega=t^2+t^{-2}$, using the technique explained in \cite{sixv} and references therein. The method is based on the observation that when $z=t^{\pm1}$, the matrix $L_n(z)$ degenerates into a projector
$$
L_n(t)=\alpha_n\otimes\beta_n^T
$$
with
$$
\alpha= \left(\begin{array}{c}1 \\ (tk_n^{-1}-t^{-1}k_n)/e_n\end{array}\right), \ \ \beta=  \left(\begin{array}{c}tk_n-t^{-1}k_n^{-1} \\ e_n\end{array}\right)
$$
We will also use the identity
\begin{align*}
L(z)L(z^{-1})&=-\det L(z)\mathrm{Id}\\
&=(\omega-z^{2}-z^{-2})\mathrm{Id}
\end{align*}
which allows us to consider the regularized reflection monodromy matrix
$$
S(z)=\left((-1)^N\prod_{i=1}^N\det L_i(z)\right)\mT(z)=L_1(z)\cdots L_N(z)\mathcal{K}_-(z)L_N(z)\cdots L_1(z)$$
which is regular at $z=t^{\pm1}$. Since $S(z)$ differs from the standard reflection monodromy matrix $\mT(z)$ by multiplication by a Casimir, the quantity $tr   \left(S(t)\mathcal{K}_+(z)\right)$ is still a reflection Hamiltonian.    
That trace may be computed as
$$
tr \left( S(t)\mathcal{K}_+(z) \right)
=(\mathcal{K}_-(t)\alpha_N,\beta_N)(\beta_1,\mathcal{K}_+(t)\alpha_1)\prod_{n=1}^{N-1}(\beta_n,\alpha_{n+1})(\beta_{n+1},\alpha_n)
$$
where we use the notation $(\alpha,\beta)$ to denote the Euclidean scalar product of the vectors $\alpha, \beta$.  On the other hand, we note that
\begin{align*}
(\beta_n,\alpha_{n+1})(\beta_{n+1},\alpha_n)&=tr \ \left(L_n(t)L_{n+1}(t) \right)\\
&=e_nf_{n+1}+f_{n+1}e_n+\omega(k_{n}k_{n+1}+k_n^{-1}k_{n+1}^{-1})-2(k_nk_{n+1}^{-1}+k_{n+1}k_{n}^{-1})
\end{align*}
and
$$
(\mathcal{K}_\pm(t)\alpha_n,\beta_n)=\text{const.}\times(\xi_\pm k_n-\xi^{-1}_\pm k_n^{-1})$$
Hence setting
\begin{align*}
H_{n,n+1}&=\log\left(e_nf_{n+1}+e_{n+1}f_n+\omega(k_nk_{n+1}+k_n^{-1}k_{n+1}^{-1})-2(k_nk_{n+1}^{-1}+k_{n+1}k_{n}^{-1})
\right)\\
H_0&=\log(\xi_+ k_1-\xi^{-1}_+ k_1^{-1})
,  \ \ H_N=\log(\xi_- k_N-\xi^{-1}_- k_N^{-1})
\end{align*}
we obtain the local reflection Hamiltonian
$$
\mathcal{H}=H_0+\sum_{n=1}^{N-1} H_{n,n+1} +H_N
$$
for the $N$-site chain.

\end{document}